\documentclass[a4]{scrartcl}
\let\origvec=\vec
\usepackage[latin1]{inputenc}
\usepackage[english]{babel}
\usepackage{amsmath}
\usepackage{amsfonts}
\usepackage{amssymb}
\usepackage{graphicx}
\usepackage{theorem}
\usepackage{tikz}
\usepgflibrary{arrows}
\usepackage[square,numbers]{natbib}
\usepackage[all]{xy}

\newtheorem{theorem}{Theorem}[section]
\newcommand{\qed}{\hfill\ensuremath{\blacksquare}}
\newenvironment{proof}{\textbf{Proof.}}{\qed}
\newcommand{\N}{\ensuremath{\mathbb{N}}}
\newcommand{\ccls}[1]{\ensuremath{\mathsf{#1}}}
\DeclareMathOperator{\NP}{\ccls{NP}}
\DeclareMathOperator{\PSPACE}{\ccls{PSPACE}}

\title{Motion planning with pull moves}
\author{Marcus Ritt\thanks{Departamento de Informática Teorica, Instituto de
  Informática,
  Universidade Federal do Rio Grande do Sul, Porto
  Alegre, Brasil}}
\date{}
\index{Ritt, Marcus}

\newif\iftrees
\treesfalse

\begin{document}
\thispagestyle{empty}
\maketitle

\begin{abstract}
  It is well known that Sokoban is
  $\ccls{PSPACE}$-complete~\citep{Culberson/1998} and several of its
  variants are $\NP$-hard~\citep{Demaine.etal/2003}. In this paper we
  prove the $\NP$-hardness of some variants of Sokoban where the
  warehouse keeper can only pull boxes.
\end{abstract}

\section{Introduction}

Sokoban is a game on an integer grid, where a warehouse keeper, or
robot, has to push boxes to designated storage locations. He can push
a box one cell horizontally or vertically if the destination cell is
free. Some of the grid cells can be occupied by unmovable obstacles.

Several variants of Sokoban have been studied. In these variants all
obstacles may be movable, the keeper may push up to $k$ boxes, or even
an unlimited number; to solve the game it can be sufficient to move
the keeper to some goal position, and the boxes once pushed may slide
until hitting an obstacle.

Table~\ref{tab:status} shows the complexity of some of these
variants. In the table the game type indicates the possible
movements. In ``Push'' games, each box moves one unit, while in
``PushPush'' a box slides until it hits the next obstacle. The
following number indicates the number of boxes which can be moved at
once, where ``*'' stands for infinity. The game type has a suffix
``-F'' when fixed obstacles are allowed. The classic Sokoban game has
game type Push-1-F.

\begin{table}[t]
  \centering
  \caption{\small Time complexity of box-moving problems. Hardness of a
    variant with all boxes movable implies hardness of the
    corresponding variant with some fixed boxes, therefore the 
    table shows only known stronger results.}
  \label{tab:status}

  \begin{minipage}{1.0\linewidth}
  \footnotesize
  \centering
  \newcommand\IY{\ensuremath{\infty}}
  \begin{tabular}{lll}
    \hline
    Game type             & Complexity& Reference\\
    \hline              
    \multicolumn{3}{c}{Path version}\\
    \hline
    Push-1                & $\NP$-hard  & \citet{Demaine.etal/2000}\\
    Push-k                & $\NP$-hard  & \citet{Demaine.etal/2003}\\
    Push-*                & $\NP$-hard  & \citet{Hoffmann/2000}\\
    PushPush-1            & $\PSPACE$-complete & \citet{Demaine.etal/2004}\\
    PushPush-k            & $\PSPACE$-complete & \citet{Demaine.etal/2004}\\
    PushPush-*            & $\NP$-hard  & \citet{Demaine.etal/2003}\\
    Push-k-F              & $\PSPACE$-complete & \citet{Demaine.etal/2002}\\
    Push-*-F              & $\PSPACE$-complete & \citet{Bremner.etal/1994}\\
    Push-k-Pull-1-F       & $\NP$-hard for $k\geq 5$ & \citet{Dor.Zwick/1999}\\
    \hline
    \multicolumn{3}{c}{Storage version}\\
    \hline
    Push-1                & $\PSPACE$-complete  & \citet{Hearn.Demaine/2005}\\
    PushPush-1            & $\NP$-hard  & \citet{Rourke.etal/1999}\\
    Push-1-F              & $\PSPACE$-complete & \citet{Culberson/1998}\\
    Push-k-Pull-1-F       & $\NP$-hard for $k\geq 5$ & \citet{Dor.Zwick/1999}\\
    \hline
  \end{tabular}
  \end{minipage}
\end{table}

\citet{Wilfong/1988} shows that motion planning where the robot can
push and pull polygonal obstacles is $\NP$-hard if we want to decide
if the robot can reach a goal position, and $\PSPACE$-complete when
goal positions for the obstacles are given. \citet{Dor.Zwick/1999}
show that Sokoban is also $\PSPACE$-complete when the robot can push
up to two boxes or pull one box, and the boxes have size $2\times
1$.
A version of Sokoban where the robot can push as well as pull the
boxes is available under the name Pukoban~\citep{Pukoban/2008}.

\citet{Polishchuk/2004} studies the optimization version of these
problems where we want to decide if there exists a solution in less
than $k$ box moves (the instances are all trivially solvable).
He shows that all of the above variations with designated storage
locations and a variation where the robot additionally can lift any
number of boxes are $\NP$-hard.

In this paper, we study the decision version of the pull-only path
variants of Sokoban (``Pull'', ``PullPull''). Different from the
pushing variants, where all adjacent boxes always have to move
together, we let the robot decide how many adjacent boxes he wants to
pull. We prove $\NP$-hardness of all Pull-$k$-F and PullPull-$k$-F and
the Pull-$1$ variants. Our proof is by reduction from planar
$3$-colorability. Its overall strategy is equal to the hardness proof
of \citet{Demaine.etal/2003} for Push-$1$. We simplify their proof by
introducing a new kind of elementary gadget (branch). In the next
section we give an overview of their approach. In
Section~\ref{sec:hardness} we show how to implement the necessary
gadgets in a pull variant of the game.

\section{Hardness of Push-$1$}

The hardness proof for Push-$1$ of \citep{Demaine.etal/2003} is by
reduction from planar $3$-coloring. The authors construct an instance
of Push-$1$, which is solvable iff a given undirected planar graph $G$
permits a $3$-coloring. The proof uses the fact that there always
exists a planar Eulerian tour $T$ in $\origvec G$, where $\origvec G$
is the directed graph obtained by substituting every undirected edge
by a pair of directed edges. The tour $T$ is augmented by decision
elements, which force the choice of a color when leaving a vertex on
its traversal. To guarantee a valid coloring, $T$ is further augmented
by consistency and coloring junctions, which force vertex colors to be
chosen consistently when a vertex is visited more than once, and
adjacent vertices to be of a different
color. Figure~\ref{fig:transformation} shows an example of a graph,
its directed version, a possible Euler tour, and the graph augmented
by decision, consistency and coloring elements.

Such a tour and its additional elements ensuring a consistent
traversal are encoded into a Push-$1$ puzzle using four kinds of basic
gadgets: (a) A one-way gadget with two entrances $A$ and $B$, that can
be traversed only in the $A$-$B$-direction. It stays open if it has
been traversed once. (b) A $1$-to-$3$-fork gadget, with four entrances
$A$ to $D$. Coming from the entrance $A$, it can be left at any other
entrance. This is choice is fixed when entering from $A$ again. (c) An
XOR-crossing gadget with an entrance on each side. It can be traversed
either from top to bottom or from left to right. This choice is fixed
in future traversals. (d) A NAND-gadget with an entrance on each
side. It can be traversed either from top to left or from bottom to
right. As for the XOR-gadget, this choice is fixed in future
traversals. The transformation substitutes the simple Eulerian path by
three parallel paths, one for each color. Figure~\ref{fig:elements}
shows the symbols for the basic and the composite elements, and
Fig.~\ref{fig:gadgets} shows how the composite elements can be
constructed from the basic ones.

\def\graphscale{1.7}
\begin{figure}
  \centering
  \begin{tikzpicture}
    [scale=\graphscale,every node/.style={fill=blue!20,circle}]
    \node (a) at (0,0) {};
    \node (b) at (1,0) {};
    \node (c) at (1,1) {};
    \node (d) at (0,1) {};
    \draw (a) to (b) to (c) to (d) to (a);
    \draw (a) to (c);
  \end{tikzpicture}\quad
  \begin{tikzpicture}
    [scale=\graphscale,every node/.style={fill=blue!20,circle}]
    \node (a) at (0,0) {};
    \node (b) at (1,0) {};
    \node (c) at (1,1) {};
    \node (d) at (0,1) {};
    \draw[->,bend left=10] (a) to (b);
    \draw[->,bend left=10] (b) to (c);
    \draw[->,bend left=10] (c) to (d);
    \draw[->,bend left=10] (d) to (a);
    \draw[->,bend left=10] (a) to (c);
    \draw[<-,bend right=10] (a) to (b);
    \draw[<-,bend right=10] (b) to (c);
    \draw[<-,bend right=10] (c) to (d);
    \draw[<-,bend right=10] (d) to (a);
    \draw[<-,bend right=10] (a) to (c);
  \end{tikzpicture}\quad
  \begin{tikzpicture}
    [scale=\graphscale,every node/.style={fill=blue!20,circle}]
    \node (a) at (0,0) {};
    \node (b) at (1,0) {};
    \node (c) at (1,1) {};
    \node (d) at (0,1) {};
    \draw[rounded corners,red] (b.south) to (a.south) to (a.west) to
    (d.west) to (d.north) to (c.north) to (c.east) to
    (b.north east) to (b.north west) to (c.south) to (a.east) to
    (a.north east) to (c.south west) to (d.south east) to (a.center)
    to (b.center) to (b.south);
  \end{tikzpicture}\\[0.5cm]
  \includegraphics[width=0.6\linewidth]{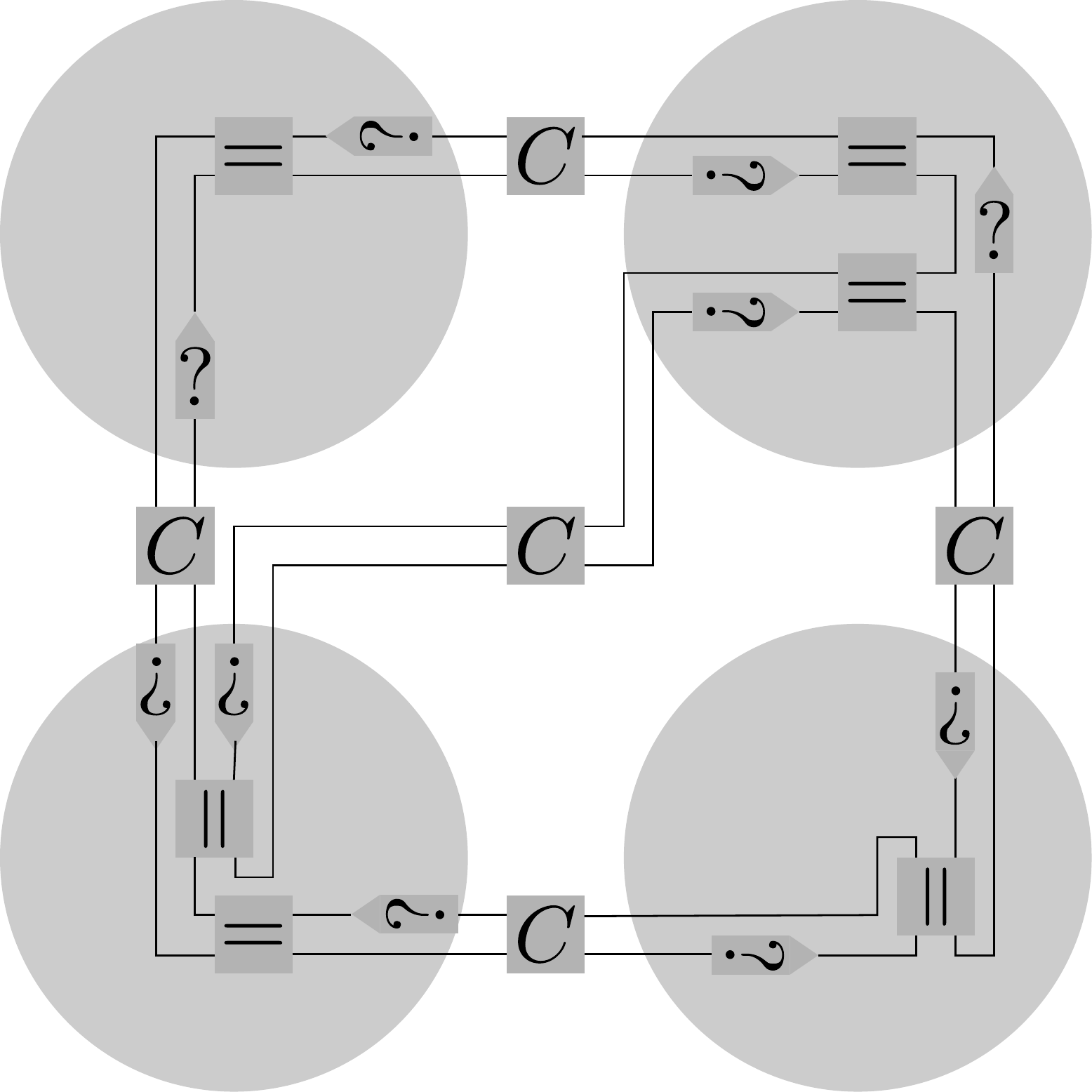}
  \caption{\small Example of a graph, its directed version, a planar Euler
    tour and the transformed tour.}
  \label{fig:transformation}
\end{figure}
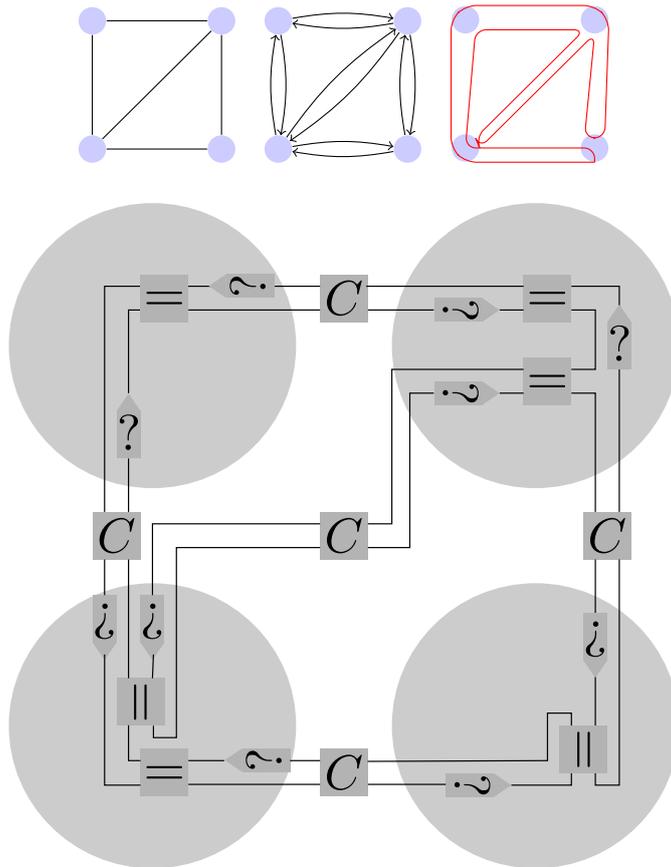

\def\gadgetscale{0.75}
\begin{figure}
  \centering
  \includegraphics[scale=\gadgetscale]{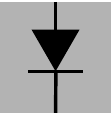}
  \includegraphics[scale=\gadgetscale]{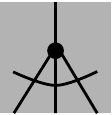}
  \includegraphics[scale=\gadgetscale]{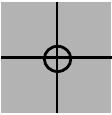}
  \includegraphics[scale=\gadgetscale]{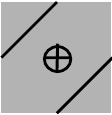}\\[0.25cm]
  \includegraphics[scale=\gadgetscale]{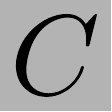}
  \includegraphics[scale=\gadgetscale]{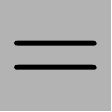}
  \includegraphics[scale=\gadgetscale]{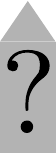}
  \caption{\small Gadgets used in the hardness proof
    of~\citet{Demaine.etal/2003}. Upper row: Basic gadgets, from left
    to right: one-way gadget, 1-to-3 fork, XOR crossing gadget, NAND
    gadget. Lower row: Composite gadgets, from left to right: coloring
    gadget, consistency gadget, color choosing gadget.}
  \label{fig:elements}
\end{figure}

\begin{figure}
  \centering
  \includegraphics[width=0.6\linewidth]{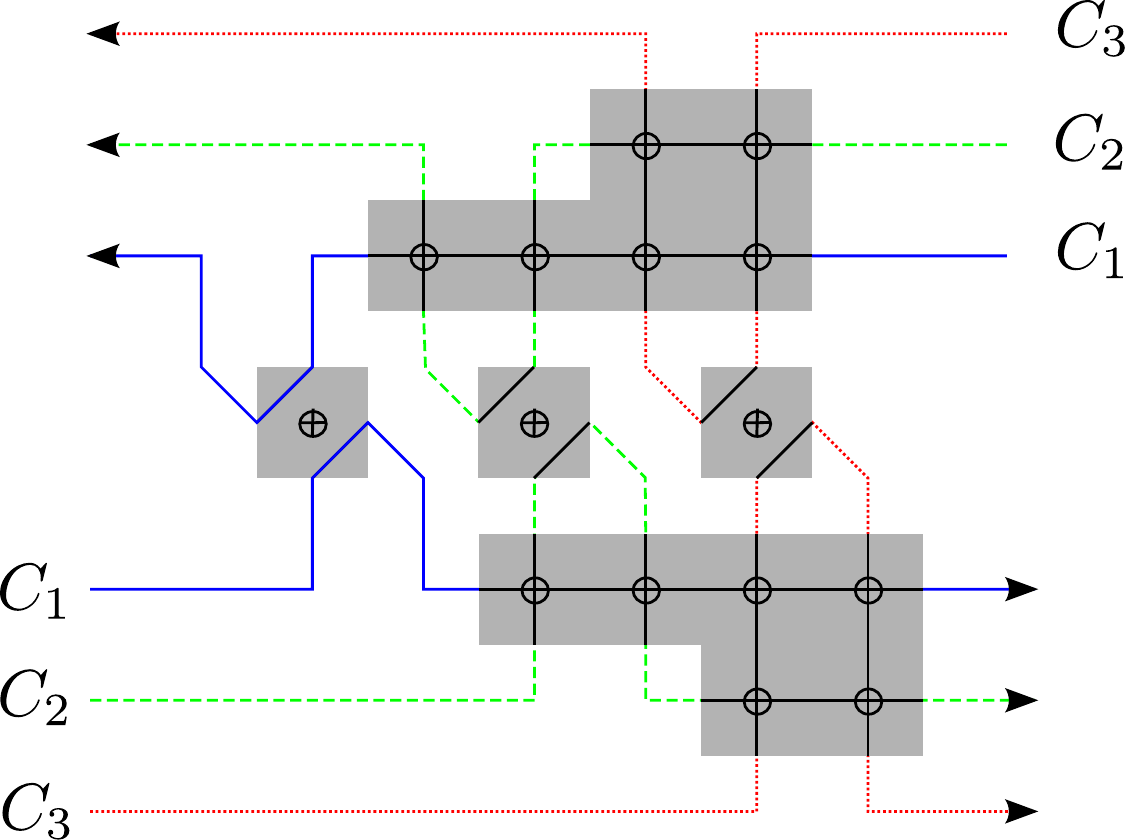}\qquad
  \includegraphics[width=0.85\linewidth]{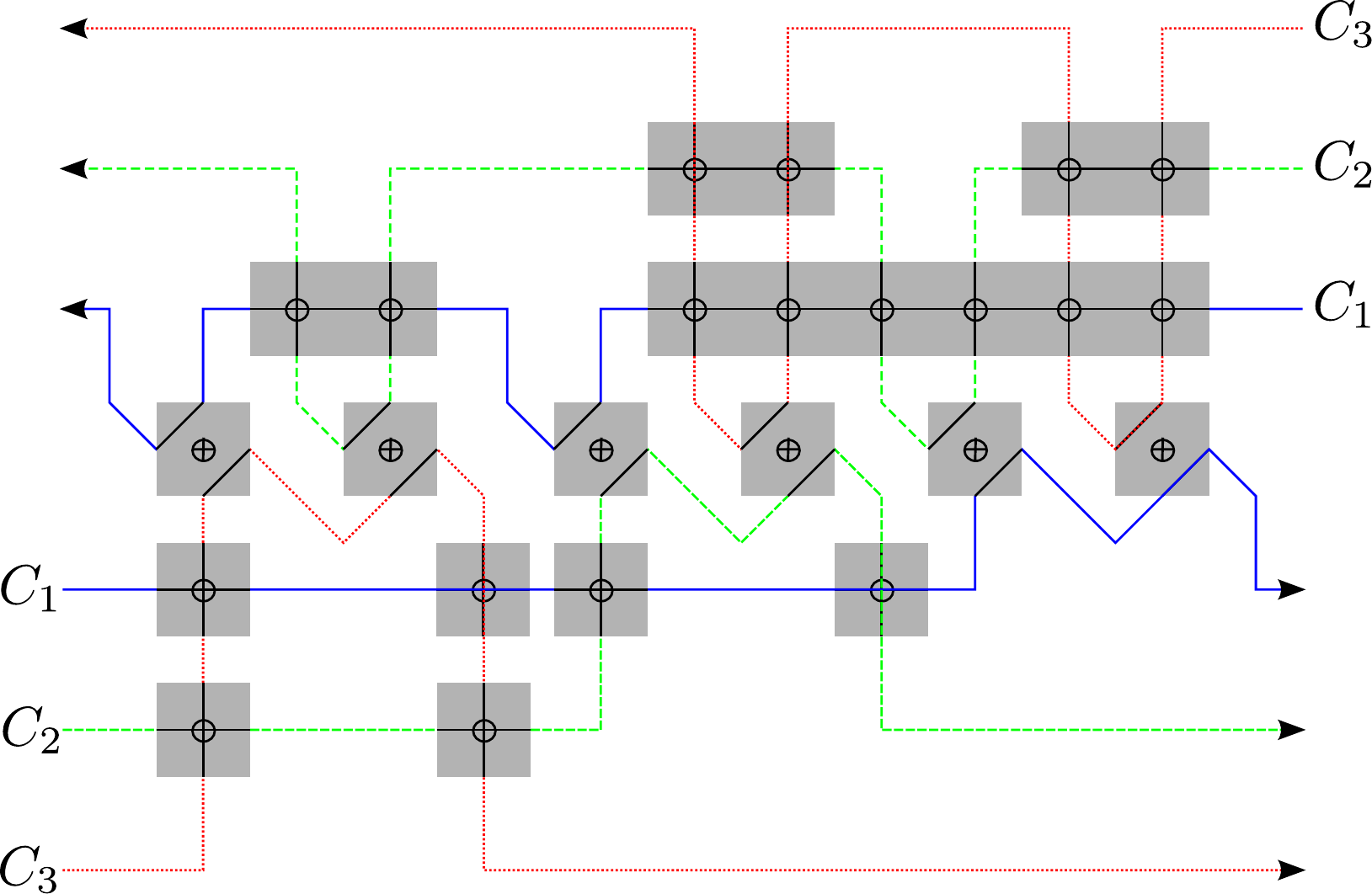}\\[0.5cm]
  \includegraphics[width=0.4\linewidth]{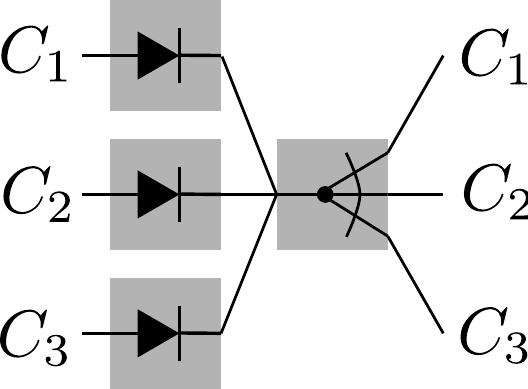}
  \caption{\small Implementations of composite gadgets. Top: Coloring
    gadget, traversable in both directions using differently labeled
    paths. Middle: consistency gadget, traversable in both
    direction only using path of the same label. Bottom: color choosing
    gadget~\citep{Demaine.etal/2003}.}
  \label{fig:gadgets}
\end{figure}

\section{Hardness of Pull-$1$-F}
\label{sec:hardness}

In this section we show that Pull-$1$-F is $\NP$-hard by giving
``pull''-implementations of the basic gadgets introduced in the
previous section. We first simplify the constructions by introducing a
branch gadget shown in Figure~\ref{fig:3xor}. It can be traversed from
$A$ to $B$ by pulling box $1$ down, permanently blocking exit $C$, or
from $B$ to either $A$ or $C$, first pulling box $1$ three positions
to the left, and then box $2$ up, if necessary, leaving all paths
permanently open. Entering from $C$ is not possible.

\begin{figure}
  \centering
  {\includegraphics[origin=c,angle=270]{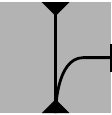}}\qquad
\begin{tikzpicture}[scale=0.3]
\draw[gray] (0,0) grid (10,8);
\filldraw[fill=blue!50,draw=black] (0,7) rectangle (1,8);
\filldraw[fill=blue!50,draw=black] (1,7) rectangle (2,8);
\filldraw[fill=blue!50,draw=black] (2,7) rectangle (3,8);
\filldraw[fill=blue!50,draw=black] (3,7) rectangle (4,8);
\filldraw[fill=blue!50,draw=black] (4,7) rectangle (5,8);
\filldraw[fill=blue!50,draw=black] (5,7) rectangle (6,8);
\filldraw[fill=blue!50,draw=black] (6,7) rectangle (7,8);
\filldraw[fill=blue!50,draw=black] (7,7) rectangle (8,8);
\filldraw[fill=blue!50,draw=black] (8,7) rectangle (9,8);
\filldraw[fill=blue!50,draw=black] (9,7) rectangle (10,8);
\filldraw[fill=blue!50,draw=black] (0,6) rectangle (1,7);
\filldraw[fill=blue!50,draw=black] (1,6) rectangle (2,7);
\filldraw[fill=blue!50,draw=black] (2,6) rectangle (3,7);
\filldraw[fill=blue!50,draw=black] (3,6) rectangle (4,7);
\filldraw[fill=blue!50,draw=black] (4,6) rectangle (5,7);
\filldraw[fill=blue!50,draw=black] (5,6) rectangle (6,7);
\filldraw[fill=blue!50,draw=black] (6,6) rectangle (7,7);
\filldraw[fill=blue!50,draw=black] (7,6) rectangle (8,7);
\filldraw[fill=blue!50,draw=black] (8,6) rectangle (9,7);
\filldraw[fill=blue!50,draw=black] (9,6) rectangle (10,7);
\filldraw[fill=blue!50,draw=black] (5,5) rectangle (6,6);
\filldraw[fill=blue!50,draw=black] (6,5) rectangle (7,6);
\filldraw[fill=blue!50,draw=black] (7,5) rectangle (8,6);
\filldraw[fill=blue!50,draw=black] (8,5) rectangle (9,6);
\filldraw[fill=blue!50,draw=black] (9,5) rectangle (10,6);
\filldraw[fill=blue!50,draw=black] (0,4) rectangle (1,5);
\filldraw[fill=blue!50,draw=black] (1,4) rectangle (2,5);
\filldraw[fill=yellow,draw=black] (6,4) rectangle (7,5);
\filldraw[fill=blue!50,draw=black] (8,4) rectangle (9,5);
\filldraw[fill=blue!50,draw=black] (9,4) rectangle (10,5);
\filldraw[fill=blue!50,draw=black] (0,3) rectangle (1,4);
\filldraw[fill=blue!50,draw=black] (1,3) rectangle (2,4);
\filldraw[fill=blue!50,draw=black] (2,3) rectangle (3,4);
\filldraw[fill=blue!50,draw=black] (3,3) rectangle (4,4);
\filldraw[fill=blue!50,draw=black] (4,3) rectangle (5,4);
\filldraw[fill=blue!50,draw=black] (5,3) rectangle (6,4);
\filldraw[fill=blue!50,draw=black] (8,3) rectangle (9,4);
\filldraw[fill=blue!50,draw=black] (9,3) rectangle (10,4);
\filldraw[fill=blue!50,draw=black] (0,2) rectangle (1,3);
\filldraw[fill=blue!50,draw=black] (1,2) rectangle (2,3);
\filldraw[fill=blue!50,draw=black] (2,2) rectangle (3,3);
\filldraw[fill=blue!50,draw=black] (3,2) rectangle (4,3);
\filldraw[fill=blue!50,draw=black] (4,2) rectangle (5,3);
\filldraw[fill=blue!50,draw=black] (5,2) rectangle (6,3);
\filldraw[fill=blue!50,draw=black] (0,1) rectangle (1,2);
\filldraw[fill=blue!50,draw=black] (1,1) rectangle (2,2);
\filldraw[fill=blue!50,draw=black] (2,1) rectangle (3,2);
\filldraw[fill=blue!50,draw=black] (3,1) rectangle (4,2);
\filldraw[fill=blue!50,draw=black] (4,1) rectangle (5,2);
\filldraw[fill=blue!50,draw=black] (5,1) rectangle (6,2);
\filldraw[fill=yellow,draw=black] (6,1) rectangle (7,2);
\filldraw[fill=blue!50,draw=black] (7,1) rectangle (8,2);
\filldraw[fill=blue!50,draw=black] (8,1) rectangle (9,2);
\filldraw[fill=blue!50,draw=black] (9,1) rectangle (10,2);
\filldraw[fill=blue!50,draw=black] (0,0) rectangle (1,1);
\filldraw[fill=blue!50,draw=black] (1,0) rectangle (2,1);
\filldraw[fill=blue!50,draw=black] (2,0) rectangle (3,1);
\filldraw[fill=blue!50,draw=black] (3,0) rectangle (4,1);
\filldraw[fill=blue!50,draw=black] (4,0) rectangle (5,1);
\filldraw[fill=blue!50,draw=black] (5,0) rectangle (6,1);
\filldraw[fill=blue!50,draw=black] (7,0) rectangle (8,1);
\filldraw[fill=blue!50,draw=black] (8,0) rectangle (9,1);
\filldraw[fill=blue!50,draw=black] (9,0) rectangle (10,1);
\node at (9.5,2.5) { $A$ };
\node at (0.5,5.5) { $B$ };
\node at (6.5,0.5) { $C$ };
\node at (6.5,4.5) { $1$ };
\node at (6.5,1.5) { $2$ };
\end{tikzpicture}

\caption{\small The branch gadget. Traversal from $A$ to $B$ is possible,
  blocking a future exit to $C$, or traversal from $B$ to $A$ or $C$
  is possible.}
  \label{fig:3xor}
\end{figure}

\begin{figure}
  \centering
  \includegraphics[scale=0.6]{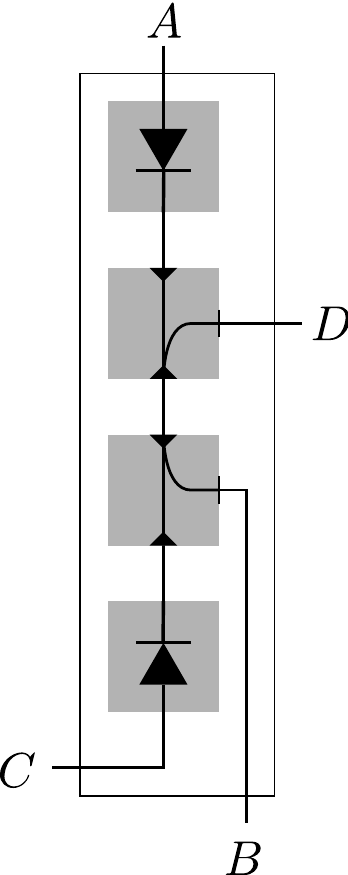}\qquad
  \includegraphics[scale=0.6]{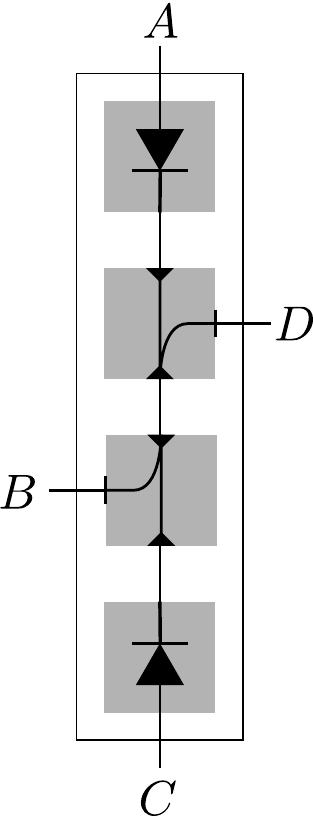}
  \caption{\small Left: XOR gadget. Right: NAND gadget.}
  \label{fig:xornand}
\end{figure}

The branch gadget can be used to implement the XOR and NAND gadgets,
as shown in Figure~\ref{fig:xornand}. The XOR-gadget
(Figure~\ref{fig:xornand}, left) can be traversed from $A$ to $B$ or
from $C$ to $D$. Traversal from $A$ to $D$ is impossible, since there
is not enough space to unblock the path to $D$ inside the branch
gadget. The exit to $C$ is protected by a one-way gadget. Similarly,
traversal from $C$ to $B$ or $A$ is not possible.

The NAND-gadget is shown in Figure~\ref{fig:xornand} (right). It can
be traversed from $A$ to $B$ or from $C$ to $D$. Traversing one path
blocks a future traversal of the other. The implementation of the
one-way gadget is shown in
Figure~\ref{fig:onewayfork}(a)\footnote{Different from the
  push-implementation in \citet{Demaine.etal/2003} a traversal from
  $B$ to $A$ is impossible. For the proof it is admissible that after
  a traversal from $A$ to $B$ the gadget stays open.}. The
$1$-to-$3$-fork gadget is implemented by cascading two $1$-to-$2$-fork
gadgets, shown in Figure~\ref{fig:onewayfork}(b). Entering at $A$ the
robot selects an exit by pulling either box $1$ first right and then
up or the box $2$ left and up, a choice which cannot be undone
entering again from $B$ or $C$.

\begin{figure}
  \centering
   \begin{tabular*}{0.45\linewidth}{c}
\begin{tikzpicture}[scale=0.3]
\draw[gray] (0,0) grid (10,9);
\filldraw[fill=blue!50,draw=black] (0,8) rectangle (1,9);
\filldraw[fill=blue!50,draw=black] (1,8) rectangle (2,9);
\filldraw[fill=blue!50,draw=black] (2,8) rectangle (3,9);
\filldraw[fill=blue!50,draw=black] (3,8) rectangle (4,9);
\filldraw[fill=blue!50,draw=black] (4,8) rectangle (5,9);
\filldraw[fill=blue!50,draw=black] (5,8) rectangle (6,9);
\filldraw[fill=blue!50,draw=black] (6,8) rectangle (7,9);
\filldraw[fill=blue!50,draw=black] (7,8) rectangle (8,9);
\filldraw[fill=blue!50,draw=black] (8,8) rectangle (9,9);
\filldraw[fill=blue!50,draw=black] (9,8) rectangle (10,9);
\filldraw[fill=blue!50,draw=black] (0,7) rectangle (1,8);
\filldraw[fill=blue!50,draw=black] (1,7) rectangle (2,8);
\filldraw[fill=blue!50,draw=black] (2,7) rectangle (3,8);
\filldraw[fill=blue!50,draw=black] (3,7) rectangle (4,8);
\filldraw[fill=blue!50,draw=black] (4,7) rectangle (5,8);
\filldraw[fill=blue!50,draw=black] (5,7) rectangle (6,8);
\filldraw[fill=blue!50,draw=black] (6,7) rectangle (7,8);
\filldraw[fill=blue!50,draw=black] (7,7) rectangle (8,8);
\filldraw[fill=blue!50,draw=black] (8,7) rectangle (9,8);
\filldraw[fill=blue!50,draw=black] (9,7) rectangle (10,8);
\filldraw[fill=blue!50,draw=black] (0,6) rectangle (1,7);
\filldraw[fill=blue!50,draw=black] (1,6) rectangle (2,7);
\filldraw[fill=blue!50,draw=black] (2,6) rectangle (3,7);
\filldraw[fill=blue!50,draw=black] (3,6) rectangle (4,7);
\filldraw[fill=blue!50,draw=black] (4,6) rectangle (5,7);
\filldraw[fill=blue!50,draw=black] (5,6) rectangle (6,7);
\filldraw[fill=blue!50,draw=black] (0,5) rectangle (1,6);
\filldraw[fill=blue!50,draw=black] (1,5) rectangle (2,6);
\filldraw[fill=blue!50,draw=black] (2,5) rectangle (3,6);
\filldraw[fill=blue!50,draw=black] (3,5) rectangle (4,6);
\filldraw[fill=blue!50,draw=black] (4,5) rectangle (5,6);
\filldraw[fill=blue!50,draw=black] (5,5) rectangle (6,6);
\filldraw[fill=blue!50,draw=black] (7,5) rectangle (8,6);
\filldraw[fill=blue!50,draw=black] (8,5) rectangle (9,6);
\filldraw[fill=blue!50,draw=black] (9,5) rectangle (10,6);
\filldraw[fill=blue!50,draw=black] (0,4) rectangle (1,5);
\filldraw[fill=blue!50,draw=black] (1,4) rectangle (2,5);
\filldraw[fill=blue!50,draw=black] (2,4) rectangle (3,5);
\filldraw[fill=blue!50,draw=black] (3,4) rectangle (4,5);
\filldraw[fill=blue!50,draw=black] (4,4) rectangle (5,5);
\filldraw[fill=blue!50,draw=black] (5,4) rectangle (6,5);
\filldraw[fill=blue!50,draw=black] (7,4) rectangle (8,5);
\filldraw[fill=blue!50,draw=black] (8,4) rectangle (9,5);
\filldraw[fill=blue!50,draw=black] (9,4) rectangle (10,5);
\filldraw[fill=blue!50,draw=black] (0,3) rectangle (1,4);
\filldraw[fill=blue!50,draw=black] (1,3) rectangle (2,4);
\filldraw[fill=yellow,draw=black] (5,3) rectangle (6,4);
\filldraw[fill=blue!50,draw=black] (8,3) rectangle (9,4);
\filldraw[fill=blue!50,draw=black] (9,3) rectangle (10,4);
\filldraw[fill=blue!50,draw=black] (0,2) rectangle (1,3);
\filldraw[fill=blue!50,draw=black] (1,2) rectangle (2,3);
\filldraw[fill=blue!50,draw=black] (3,2) rectangle (4,3);
\filldraw[fill=blue!50,draw=black] (4,2) rectangle (5,3);
\filldraw[fill=blue!50,draw=black] (8,2) rectangle (9,3);
\filldraw[fill=blue!50,draw=black] (9,2) rectangle (10,3);
\filldraw[fill=blue!50,draw=black] (0,1) rectangle (1,2);
\filldraw[fill=blue!50,draw=black] (1,1) rectangle (2,2);
\filldraw[fill=blue!50,draw=black] (3,1) rectangle (4,2);
\filldraw[fill=blue!50,draw=black] (4,1) rectangle (5,2);
\filldraw[fill=blue!50,draw=black] (5,1) rectangle (6,2);
\filldraw[fill=blue!50,draw=black] (6,1) rectangle (7,2);
\filldraw[fill=blue!50,draw=black] (7,1) rectangle (8,2);
\filldraw[fill=blue!50,draw=black] (8,1) rectangle (9,2);
\filldraw[fill=blue!50,draw=black] (9,1) rectangle (10,2);
\filldraw[fill=blue!50,draw=black] (0,0) rectangle (1,1);
\filldraw[fill=blue!50,draw=black] (1,0) rectangle (2,1);
\filldraw[fill=blue!50,draw=black] (3,0) rectangle (4,1);
\filldraw[fill=blue!50,draw=black] (4,0) rectangle (5,1);
\filldraw[fill=blue!50,draw=black] (5,0) rectangle (6,1);
\filldraw[fill=blue!50,draw=black] (6,0) rectangle (7,1);
\filldraw[fill=blue!50,draw=black] (7,0) rectangle (8,1);
\filldraw[fill=blue!50,draw=black] (8,0) rectangle (9,1);
\filldraw[fill=blue!50,draw=black] (9,0) rectangle (10,1);
\node at (2.5,0.5) { $B$ };
\node at (9.5,6.5) { $A$ };
\end{tikzpicture}\\(a)
   \end{tabular*}
   \begin{tabular*}{0.45\linewidth}{c}
\begin{tikzpicture}[scale=0.3]
\draw[gray] (0,0) grid (12,8);
\filldraw[fill=blue!50,draw=black] (0,7) rectangle (1,8);
\filldraw[fill=blue!50,draw=black] (1,7) rectangle (2,8);
\filldraw[fill=blue!50,draw=black] (3,7) rectangle (4,8);
\filldraw[fill=blue!50,draw=black] (4,7) rectangle (5,8);
\filldraw[fill=blue!50,draw=black] (5,7) rectangle (6,8);
\filldraw[fill=blue!50,draw=black] (6,7) rectangle (7,8);
\filldraw[fill=blue!50,draw=black] (7,7) rectangle (8,8);
\filldraw[fill=blue!50,draw=black] (8,7) rectangle (9,8);
\filldraw[fill=blue!50,draw=black] (10,7) rectangle (11,8);
\filldraw[fill=blue!50,draw=black] (11,7) rectangle (12,8);
\filldraw[fill=blue!50,draw=black] (0,6) rectangle (1,7);
\filldraw[fill=blue!50,draw=black] (1,6) rectangle (2,7);
\filldraw[fill=blue!50,draw=black] (3,6) rectangle (4,7);
\filldraw[fill=blue!50,draw=black] (4,6) rectangle (5,7);
\filldraw[fill=blue!50,draw=black] (5,6) rectangle (6,7);
\filldraw[fill=blue!50,draw=black] (6,6) rectangle (7,7);
\filldraw[fill=blue!50,draw=black] (7,6) rectangle (8,7);
\filldraw[fill=blue!50,draw=black] (8,6) rectangle (9,7);
\filldraw[fill=blue!50,draw=black] (10,6) rectangle (11,7);
\filldraw[fill=blue!50,draw=black] (11,6) rectangle (12,7);
\filldraw[fill=blue!50,draw=black] (0,5) rectangle (1,6);
\filldraw[fill=blue!50,draw=black] (1,5) rectangle (2,6);
\filldraw[fill=blue!50,draw=black] (3,5) rectangle (4,6);
\filldraw[fill=blue!50,draw=black] (4,5) rectangle (5,6);
\filldraw[fill=blue!50,draw=black] (7,5) rectangle (8,6);
\filldraw[fill=blue!50,draw=black] (8,5) rectangle (9,6);
\filldraw[fill=blue!50,draw=black] (10,5) rectangle (11,6);
\filldraw[fill=blue!50,draw=black] (11,5) rectangle (12,6);
\filldraw[fill=blue!50,draw=black] (0,4) rectangle (1,5);
\filldraw[fill=blue!50,draw=black] (1,4) rectangle (2,5);
\filldraw[fill=blue!50,draw=black] (3,4) rectangle (4,5);
\filldraw[fill=blue!50,draw=black] (4,4) rectangle (5,5);
\filldraw[fill=blue!50,draw=black] (7,4) rectangle (8,5);
\filldraw[fill=blue!50,draw=black] (8,4) rectangle (9,5);
\filldraw[fill=blue!50,draw=black] (10,4) rectangle (11,5);
\filldraw[fill=blue!50,draw=black] (11,4) rectangle (12,5);
\filldraw[fill=blue!50,draw=black] (0,3) rectangle (1,4);
\filldraw[fill=blue!50,draw=black] (1,3) rectangle (2,4);
\filldraw[fill=yellow,draw=black] (4,3) rectangle (5,4);
\filldraw[fill=yellow,draw=black] (7,3) rectangle (8,4);
\filldraw[fill=blue!50,draw=black] (10,3) rectangle (11,4);
\filldraw[fill=blue!50,draw=black] (11,3) rectangle (12,4);
\filldraw[fill=blue!50,draw=black] (0,2) rectangle (1,3);
\filldraw[fill=blue!50,draw=black] (1,2) rectangle (2,3);
\filldraw[fill=blue!50,draw=black] (2,2) rectangle (3,3);
\filldraw[fill=blue!50,draw=black] (3,2) rectangle (4,3);
\filldraw[fill=blue!50,draw=black] (4,2) rectangle (5,3);
\filldraw[fill=blue!50,draw=black] (6,2) rectangle (7,3);
\filldraw[fill=blue!50,draw=black] (7,2) rectangle (8,3);
\filldraw[fill=blue!50,draw=black] (8,2) rectangle (9,3);
\filldraw[fill=blue!50,draw=black] (9,2) rectangle (10,3);
\filldraw[fill=blue!50,draw=black] (10,2) rectangle (11,3);
\filldraw[fill=blue!50,draw=black] (11,2) rectangle (12,3);
\filldraw[fill=blue!50,draw=black] (0,1) rectangle (1,2);
\filldraw[fill=blue!50,draw=black] (1,1) rectangle (2,2);
\filldraw[fill=blue!50,draw=black] (2,1) rectangle (3,2);
\filldraw[fill=blue!50,draw=black] (3,1) rectangle (4,2);
\filldraw[fill=blue!50,draw=black] (4,1) rectangle (5,2);
\filldraw[fill=blue!50,draw=black] (6,1) rectangle (7,2);
\filldraw[fill=blue!50,draw=black] (7,1) rectangle (8,2);
\filldraw[fill=blue!50,draw=black] (8,1) rectangle (9,2);
\filldraw[fill=blue!50,draw=black] (9,1) rectangle (10,2);
\filldraw[fill=blue!50,draw=black] (10,1) rectangle (11,2);
\filldraw[fill=blue!50,draw=black] (11,1) rectangle (12,2);
\filldraw[fill=blue!50,draw=black] (0,0) rectangle (1,1);
\filldraw[fill=blue!50,draw=black] (1,0) rectangle (2,1);
\filldraw[fill=blue!50,draw=black] (2,0) rectangle (3,1);
\filldraw[fill=blue!50,draw=black] (3,0) rectangle (4,1);
\filldraw[fill=blue!50,draw=black] (4,0) rectangle (5,1);
\filldraw[fill=blue!50,draw=black] (6,0) rectangle (7,1);
\filldraw[fill=blue!50,draw=black] (7,0) rectangle (8,1);
\filldraw[fill=blue!50,draw=black] (8,0) rectangle (9,1);
\filldraw[fill=blue!50,draw=black] (9,0) rectangle (10,1);
\filldraw[fill=blue!50,draw=black] (10,0) rectangle (11,1);
\filldraw[fill=blue!50,draw=black] (11,0) rectangle (12,1);
\node at (5.5,0.5) { $A$ };
\node at (2.5,7.5) { $B$ };
\node at (9.5,7.5) { $C$ };
\node at (4.5,3.5) { $1$ };
\node at (7.5,3.5) { $2$ };
\end{tikzpicture}\\(b)

   \end{tabular*}
   \caption{\small(a) One-way gadget permitting traversal from $A$ to
     $B$. (b) $1$-to-$2$-fork gadget. A traversal from $A$ to either
     $B$ or $C$ is possible.}
  \label{fig:onewayfork}
\end{figure}

From the ``pull-implementations'' of the elementary gadgets above and
the proof of \citet{Demaine.etal/2003} we obtain:
\begin{theorem}
  Pull-$1$-F is $\NP$-hard.
\end{theorem}

\section{Hardness of variants}

The aobve constructions also apply to pulling several boxes.

\begin{theorem}
  Pull-$k$-F is $\NP$-hard for any $k\in\N\cup\{\mathrm{*}\}$.
\end{theorem}

\begin{proof}
  For arbitrary $k$, the one-way gadget continues to work the same,
  since it has only one movable obstacle. For the $1$-to-$2$-fork and
  the branch gadget, the two movable obstacles cannot be brought
  together in such a way that the robot can move both obstacles. Thus,
  the additional power for $k>2$ or $k=\mathrm{*}$ does not help to
  solve the puzzle.
\end{proof}

  All gadgets continue to work under PullPull conditions. All the
  necessary moves for the permitted traversals are already PullPull
  moves, and the restriction to PullPull moves clearly does not open
  new paths. We have therefore:

  \begin{theorem}
    PullPull-$k$-F is $\NP$-hard for any $k\in\N$.
  \end{theorem}

  We next turn to the simplest case without fixed boxes.

  \begin{theorem}
    Pull-$1$ is $\NP$-hard.
  \end{theorem}
  \begin{proof}
    We first observe that the robot moves only inside the gadgets. We
    can make sure that it cannot escape from there, even if all
    obstacles are movable, making the surrounding walls thicker than
    the interior space. Figure~\ref{fig:ann1} shows for each gadget
    all movable blocks when entering from any point.

    Valid passages from the cases with fixed obstacles stay valid, so
    we need only to show that moving boxes that were previously
    unmovable do not open new passages.
  
    It is easy to verify that a large number of the moves of such
    boxes only reduce the available space for the robot, or even
    capture him in a hole. For example, when entering the one-way
    gadget from $B$, there are only two movable boxes at the first
    corner. Moving one of these locks the robot in a part of the
    corridor, but does not open a passage to $A$. Similarly, it is not
    difficult to verify that we cannot pass the fork gadgets from $B$
    or $C$, or the branch gadget coming from $C$. It remains to verify
    that after entering the fork gadget from $A$ we can leave either
    at $B$ or $C$, but not both, and that after entering the branch
    gadget from $A$ we can exit only at $B$, and we can do this only
    by blocking $C$ permanently. This has been done by a simulation of
    all possible movements.
\end{proof}

\def\mboxscale{1}
\def\mmboxscale{0.27}

\begin{figure}
  \centering
  \scalebox{\mboxscale}{%
\begin{tikzpicture}[scale=\mmboxscale]
\draw[gray] (0,0) grid (10,9);
\filldraw[fill=blue!50,draw=black] (0,8) rectangle (1,9);
\filldraw[fill=blue!50,draw=black] (1,8) rectangle (2,9);
\filldraw[fill=blue!50,draw=black] (2,8) rectangle (3,9);
\filldraw[fill=blue!50,draw=black] (3,8) rectangle (4,9);
\filldraw[fill=blue!50,draw=black] (4,8) rectangle (5,9);
\filldraw[fill=blue!50,draw=black] (5,8) rectangle (6,9);
\filldraw[fill=blue!50,draw=black] (6,8) rectangle (7,9);
\filldraw[fill=blue!50,draw=black] (7,8) rectangle (8,9);
\filldraw[fill=blue!50,draw=black] (8,8) rectangle (9,9);
\filldraw[fill=blue!50,draw=black] (9,8) rectangle (10,9);
\filldraw[fill=blue!50,draw=black] (0,7) rectangle (1,8);
\filldraw[fill=blue!50,draw=black] (1,7) rectangle (2,8);
\filldraw[fill=blue!50,draw=black] (2,7) rectangle (3,8);
\filldraw[fill=blue!50,draw=black] (3,7) rectangle (4,8);
\filldraw[fill=blue!50,draw=black] (4,7) rectangle (5,8);
\filldraw[fill=blue!50,draw=black] (5,7) rectangle (6,8);
\filldraw[fill=blue!50,draw=black] (6,7) rectangle (7,8);
\filldraw[fill=blue!50,draw=black] (7,7) rectangle (8,8);
\filldraw[fill=blue!50,draw=black] (8,7) rectangle (9,8);
\filldraw[fill=blue!50,draw=black] (9,7) rectangle (10,8);
\filldraw[fill=blue!50,draw=black] (0,6) rectangle (1,7);
\filldraw[fill=blue!50,draw=black] (1,6) rectangle (2,7);
\filldraw[fill=blue!50,draw=black] (2,6) rectangle (3,7);
\filldraw[fill=blue!50,draw=black] (3,6) rectangle (4,7);
\filldraw[fill=blue!50,draw=black] (4,6) rectangle (5,7);
\filldraw[fill=blue!50,draw=black] (5,6) rectangle (6,7);
\filldraw[fill=blue!50,draw=black] (0,5) rectangle (1,6);
\filldraw[fill=blue!50,draw=black] (1,5) rectangle (2,6);
\filldraw[fill=blue!50,draw=black] (2,5) rectangle (3,6);
\filldraw[fill=blue!50,draw=black] (3,5) rectangle (4,6);
\filldraw[fill=blue!50,draw=black] (4,5) rectangle (5,6);
\filldraw[fill=blue!50,draw=black] (5,5) rectangle (6,6);
\filldraw[fill=blue!50,draw=black] (7,5) rectangle (8,6);
\filldraw[fill=blue!50,draw=black] (8,5) rectangle (9,6);
\filldraw[fill=blue!50,draw=black] (9,5) rectangle (10,6);
\filldraw[fill=blue!50,draw=black] (0,4) rectangle (1,5);
\filldraw[fill=blue!50,draw=black] (1,4) rectangle (2,5);
\filldraw[fill=blue!50,draw=black] (2,4) rectangle (3,5);
\filldraw[fill=blue!50,draw=black] (3,4) rectangle (4,5);
\filldraw[fill=blue!50,draw=black] (4,4) rectangle (5,5);
\filldraw[fill=blue!50,draw=black] (5,4) rectangle (6,5);
\filldraw[fill=blue!50,draw=black] (7,4) rectangle (8,5);
\filldraw[fill=blue!50,draw=black] (8,4) rectangle (9,5);
\filldraw[fill=blue!50,draw=black] (9,4) rectangle (10,5);
\filldraw[fill=blue!50,draw=black] (0,3) rectangle (1,4);
\filldraw[fill=blue!50,draw=black] (1,3) rectangle (2,4);
\filldraw[fill=blue!50,draw=black] (5,3) rectangle (6,4);
\filldraw[fill=blue!50,draw=black] (8,3) rectangle (9,4);
\filldraw[fill=blue!50,draw=black] (9,3) rectangle (10,4);
\filldraw[fill=blue!50,draw=black] (0,2) rectangle (1,3);
\filldraw[fill=blue!50,draw=black] (1,2) rectangle (2,3);
\filldraw[fill=blue!50,draw=black] (3,2) rectangle (4,3);
\filldraw[fill=blue!50,draw=black] (4,2) rectangle (5,3);
\filldraw[fill=blue!50,draw=black] (8,2) rectangle (9,3);
\filldraw[fill=blue!50,draw=black] (9,2) rectangle (10,3);
\filldraw[fill=blue!50,draw=black] (0,1) rectangle (1,2);
\filldraw[fill=blue!50,draw=black] (1,1) rectangle (2,2);
\filldraw[fill=blue!50,draw=black] (3,1) rectangle (4,2);
\filldraw[fill=blue!50,draw=black] (4,1) rectangle (5,2);
\filldraw[fill=blue!50,draw=black] (5,1) rectangle (6,2);
\filldraw[fill=blue!50,draw=black] (6,1) rectangle (7,2);
\filldraw[fill=blue!50,draw=black] (7,1) rectangle (8,2);
\filldraw[fill=blue!50,draw=black] (8,1) rectangle (9,2);
\filldraw[fill=blue!50,draw=black] (9,1) rectangle (10,2);
\filldraw[fill=blue!50,draw=black] (0,0) rectangle (1,1);
\filldraw[fill=blue!50,draw=black] (1,0) rectangle (2,1);
\filldraw[fill=blue!50,draw=black] (3,0) rectangle (4,1);
\filldraw[fill=blue!50,draw=black] (4,0) rectangle (5,1);
\filldraw[fill=blue!50,draw=black] (5,0) rectangle (6,1);
\filldraw[fill=blue!50,draw=black] (6,0) rectangle (7,1);
\filldraw[fill=blue!50,draw=black] (7,0) rectangle (8,1);
\filldraw[fill=blue!50,draw=black] (8,0) rectangle (9,1);
\filldraw[fill=blue!50,draw=black] (9,0) rectangle (10,1);
\iftrees
\draw[*-,red,ultra thick] (6.25,7.5) -- (6.25,6.5) -- (6.25,5.5) -- (6.25,4.5) -- (6.25,3.75);
\draw[*-,red,ultra thick] (5.5,6.5) -- (6.5,6.5) -- (7.5,6.5) -- (8.5,6.5);
\draw[*-,red,ultra thick] (2.5,4.5) -- (2.5,3.5) -- (2.5,2.5) -- (2.5,1.5);
WARNING: id 41 is no simple path!
\draw[*-,red,ultra thick] (5.4,4.5) -- (6.15,4.5) -- (5.5,4.5) -- (5.5,3.75);
\draw[*-,red,ultra thick] (7.5,4.5) -- (7.5,3.75);
\draw[*-,red,ultra thick] (8.5,4.5) -- (7.75,4.5);
\draw[*-,red,ultra thick] (1.5,3.75) -- (2.5,3.75) -- (3.5,3.75) -- (4.5,3.75);
WARNING: id 47 is no simple path!
\draw[*-,red,ultra thick] (5.4,3.5) -- (6.25,3.5) -- (5.5,3.5) -- (3.5,3.5);
\draw[*-,red,ultra thick] (8.5,3.5) -- (7.75,3.5);
\draw[*-,red,ultra thick] (4.5,2.75) -- (5.5,2.75) -- (6.5,2.75) -- (6.5,3.5) -- (6.5,4.5) -- (6.5,5.5);
\draw[*-,red,ultra thick] (8.5,2.5) -- (7.5,2.5) -- (6.5,2.5);
\draw[*-,red,ultra thick] (5.5,1.5) -- (5.5,2.5);
\draw[*-,red,ultra thick] (6.75,1.5) -- (6.75,2.5) -- (6.75,3.5) -- (6.75,4.5) -- (6.75,5.5);
\draw[*-,red,ultra thick] (7.5,1.5) -- (7.5,2.25) -- (6.5,2.25);
\else
\node[circle,fill=red,inner sep=1.5pt] at (6.5,7.5) {};
\node[circle,fill=red,inner sep=1.5pt] at (5.5,6.5)  {};
\node[circle,fill=red,inner sep=1.5pt] at (2.5,4.5)  {};
\node[circle,fill=red,inner sep=1.5pt] at (5.4,4.5)  {};
\node[circle,fill=red,inner sep=1.5pt] at (7.5,4.5)  {};
\node[circle,fill=red,inner sep=1.5pt] at (8.5,4.5)  {};
\node[circle,fill=red,inner sep=1.5pt] at (1.5,3.5) {};
\node[circle,fill=red,inner sep=1.5pt] at (5.4,3.5)  {};
\node[circle,fill=red,inner sep=1.5pt] at (8.5,3.5)  {};
\node[circle,fill=red,inner sep=1.5pt] at (4.5,2.5) {};
\node[circle,fill=red,inner sep=1.5pt] at (8.5,2.5)  {};
\node[circle,fill=red,inner sep=1.5pt] at (5.5,1.5)  {};
\node[circle,fill=red,inner sep=1.5pt] at (6.5,1.5) {};
\node[circle,fill=red,inner sep=1.5pt] at (7.5,1.5)  {};
\fi
\node at (2.5,0.5) { $B$ };
\node at (9.5,6.5) { $A$ };
\end{tikzpicture}
}\qquad
\scalebox{\mboxscale}{
\begin{tikzpicture}[scale=\mmboxscale]
\draw[gray] (0,0) grid (12,9);
\filldraw[fill=blue!50,draw=black] (0,8) rectangle (1,9);
\filldraw[fill=blue!50,draw=black] (1,8) rectangle (2,9);
\filldraw[fill=blue!50,draw=black] (3,8) rectangle (4,9);
\filldraw[fill=blue!50,draw=black] (4,8) rectangle (5,9);
\filldraw[fill=blue!50,draw=black] (5,8) rectangle (6,9);
\filldraw[fill=blue!50,draw=black] (6,8) rectangle (7,9);
\filldraw[fill=blue!50,draw=black] (7,8) rectangle (8,9);
\filldraw[fill=blue!50,draw=black] (8,8) rectangle (9,9);
\filldraw[fill=blue!50,draw=black] (10,8) rectangle (11,9);
\filldraw[fill=blue!50,draw=black] (11,8) rectangle (12,9);
\filldraw[fill=blue!50,draw=black] (0,7) rectangle (1,8);
\filldraw[fill=blue!50,draw=black] (1,7) rectangle (2,8);
\filldraw[fill=blue!50,draw=black] (3,7) rectangle (4,8);
\filldraw[fill=blue!50,draw=black] (4,7) rectangle (5,8);
\filldraw[fill=blue!50,draw=black] (5,7) rectangle (6,8);
\filldraw[fill=blue!50,draw=black] (6,7) rectangle (7,8);
\filldraw[fill=blue!50,draw=black] (7,7) rectangle (8,8);
\filldraw[fill=blue!50,draw=black] (8,7) rectangle (9,8);
\filldraw[fill=blue!50,draw=black] (10,7) rectangle (11,8);
\filldraw[fill=blue!50,draw=black] (11,7) rectangle (12,8);
\filldraw[fill=blue!50,draw=black] (0,6) rectangle (1,7);
\filldraw[fill=blue!50,draw=black] (1,6) rectangle (2,7);
\filldraw[fill=blue!50,draw=black] (3,6) rectangle (4,7);
\filldraw[fill=blue!50,draw=black] (4,6) rectangle (5,7);
\filldraw[fill=blue!50,draw=black] (5,6) rectangle (6,7);
\filldraw[fill=blue!50,draw=black] (6,6) rectangle (7,7);
\filldraw[fill=blue!50,draw=black] (7,6) rectangle (8,7);
\filldraw[fill=blue!50,draw=black] (8,6) rectangle (9,7);
\filldraw[fill=blue!50,draw=black] (10,6) rectangle (11,7);
\filldraw[fill=blue!50,draw=black] (11,6) rectangle (12,7);
\filldraw[fill=blue!50,draw=black] (0,5) rectangle (1,6);
\filldraw[fill=blue!50,draw=black] (1,5) rectangle (2,6);
\filldraw[fill=blue!50,draw=black] (3,5) rectangle (4,6);
\filldraw[fill=blue!50,draw=black] (4,5) rectangle (5,6);
\filldraw[fill=blue!50,draw=black] (7,5) rectangle (8,6);
\filldraw[fill=blue!50,draw=black] (8,5) rectangle (9,6);
\filldraw[fill=blue!50,draw=black] (10,5) rectangle (11,6);
\filldraw[fill=blue!50,draw=black] (11,5) rectangle (12,6);
\filldraw[fill=blue!50,draw=black] (0,4) rectangle (1,5);
\filldraw[fill=blue!50,draw=black] (1,4) rectangle (2,5);
\filldraw[fill=blue!50,draw=black] (3,4) rectangle (4,5);
\filldraw[fill=blue!50,draw=black] (4,4) rectangle (5,5);
\filldraw[fill=blue!50,draw=black] (7,4) rectangle (8,5);
\filldraw[fill=blue!50,draw=black] (8,4) rectangle (9,5);
\filldraw[fill=blue!50,draw=black] (10,4) rectangle (11,5);
\filldraw[fill=blue!50,draw=black] (11,4) rectangle (12,5);
\filldraw[fill=blue!50,draw=black] (0,3) rectangle (1,4);
\filldraw[fill=blue!50,draw=black] (1,3) rectangle (2,4);
\filldraw[fill=blue!50,draw=black] (4,3) rectangle (5,4);
\filldraw[fill=blue!50,draw=black] (7,3) rectangle (8,4);
\filldraw[fill=blue!50,draw=black] (10,3) rectangle (11,4);
\filldraw[fill=blue!50,draw=black] (11,3) rectangle (12,4);
\filldraw[fill=blue!50,draw=black] (0,2) rectangle (1,3);
\filldraw[fill=blue!50,draw=black] (1,2) rectangle (2,3);
\filldraw[fill=blue!50,draw=black] (2,2) rectangle (3,3);
\filldraw[fill=blue!50,draw=black] (3,2) rectangle (4,3);
\filldraw[fill=blue!50,draw=black] (4,2) rectangle (5,3);
\filldraw[fill=blue!50,draw=black] (6,2) rectangle (7,3);
\filldraw[fill=blue!50,draw=black] (7,2) rectangle (8,3);
\filldraw[fill=blue!50,draw=black] (8,2) rectangle (9,3);
\filldraw[fill=blue!50,draw=black] (9,2) rectangle (10,3);
\filldraw[fill=blue!50,draw=black] (10,2) rectangle (11,3);
\filldraw[fill=blue!50,draw=black] (11,2) rectangle (12,3);
\filldraw[fill=blue!50,draw=black] (0,1) rectangle (1,2);
\filldraw[fill=blue!50,draw=black] (1,1) rectangle (2,2);
\filldraw[fill=blue!50,draw=black] (2,1) rectangle (3,2);
\filldraw[fill=blue!50,draw=black] (3,1) rectangle (4,2);
\filldraw[fill=blue!50,draw=black] (4,1) rectangle (5,2);
\filldraw[fill=blue!50,draw=black] (6,1) rectangle (7,2);
\filldraw[fill=blue!50,draw=black] (7,1) rectangle (8,2);
\filldraw[fill=blue!50,draw=black] (8,1) rectangle (9,2);
\filldraw[fill=blue!50,draw=black] (9,1) rectangle (10,2);
\filldraw[fill=blue!50,draw=black] (10,1) rectangle (11,2);
\filldraw[fill=blue!50,draw=black] (11,1) rectangle (12,2);
\filldraw[fill=blue!50,draw=black] (0,0) rectangle (1,1);
\filldraw[fill=blue!50,draw=black] (1,0) rectangle (2,1);
\filldraw[fill=blue!50,draw=black] (2,0) rectangle (3,1);
\filldraw[fill=blue!50,draw=black] (3,0) rectangle (4,1);
\filldraw[fill=blue!50,draw=black] (4,0) rectangle (5,1);
\filldraw[fill=blue!50,draw=black] (6,0) rectangle (7,1);
\filldraw[fill=blue!50,draw=black] (7,0) rectangle (8,1);
\filldraw[fill=blue!50,draw=black] (8,0) rectangle (9,1);
\filldraw[fill=blue!50,draw=black] (9,0) rectangle (10,1);
\filldraw[fill=blue!50,draw=black] (10,0) rectangle (11,1);
\filldraw[fill=blue!50,draw=black] (11,0) rectangle (12,1);
\iftrees
\draw[*-,red,very thick] (5.5,7.5) -- (5.5,6.75);
\draw[*-,red,very thick] (6.5,7.5) -- (6.5,6.75);
\draw[*-,red,very thick] (5.5,6.5) -- (5.5,5.5) -- (5.5,4.5) -- (5.5,3.5) -- (5.5,2.5) -- (5.5,1.75);
\draw[*-,red,very thick] (6.5,6.5) -- (6.5,5.5) -- (6.5,4.75);
\draw[*-,red,very thick] (3.5,5.75) -- (4.5,5.75) -- (5.25,5.75);
\draw[*-,red,very thick] (4.5,5.5) -- (5.5,5.5) -- (6.25,5.5);
\draw[*-,red,very thick] (7.5,5.25) -- (6.75,5.25) -- (6.75,4.75);
\draw[*-,red,very thick] (8.5,5.5) -- (7.5,5.5) -- (6.75,5.5);
\draw[*-,red,very thick] (4.5,4.5) -- (5.25,4.5) -- (5.25,3.5) -- (5.25,2.5) -- (5.25,1.75);
\draw[*-,red,very thick] (7.5,4.5) -- (6.75,4.5);
\draw[*-,red,very thick] (1.5,3.75) -- (2.5,3.75) -- (3.5,3.75) -- (4.5,3.75) -- (5.15,3.75);
\draw[*-,red,very thick] (4.5,3.5) -- (5.5,3.5) -- (6.5,3.5) -- (7.5,3.5) -- (8.25,3.5);
\draw[-,red,very thick] (4.5,3.5) -- (3.5,3.5); 
\draw[*-,red,very thick] (7.5,3.25) -- (6.5,3.25) -- (5.5,3.25) -- (4.5,3.25) -- (3.75,3.25);
\draw[-,red,very thick] (7.5,3.25) -- (8.5,3.25); 
\draw[*-,red,very thick] (10.5,3.75) -- (9.5,3.75) -- (8.5,3.75) -- (7.5,3.75) -- (6.75,3.75);
\draw[*-,red,very thick] (2.5,2.5) -- (2.5,3.5) -- (2.5,4.5) -- (2.5,5.5) -- (2.5,6.5) -- (2.5,7.25);
\draw[*-,red,very thick] (4.5,2.5) -- (5.75,2.5) -- (5.75,3.5) -- (5.75,4.25);
\draw[*-,red,very thick] (6.5,2.5) -- (6.5,3.5) -- (6.5,4.25);
\draw[*-,red,very thick] (7.5,2.5) -- (6.75,2.5);
\draw[*-,red,very thick] (9.5,2.5) -- (9.5,3.5) -- (9.5,4.5) -- (9.5,5.5) -- (9.5,6.5) -- (9.5,7.25);
\draw[*-,red,very thick] (4.5,1.5) -- (5.25,1.5);
\draw[*-,red,very thick] (6.5,1.5) -- (6.5,2.25);
\draw[*-,red,very thick] (7.5,1.5) -- (6.75,1.5);
\else
\node[circle,fill=red,inner sep=1.5pt] at (5.5,7.5)  {};
\node[circle,fill=red,inner sep=1.5pt] at (6.5,7.5)  {};
\node[circle,fill=red,inner sep=1.5pt] at (5.5,6.5)  {};
\node[circle,fill=red,inner sep=1.5pt] at (6.5,6.5)  {};
\node[circle,fill=red,inner sep=1.5pt] at (3.5,5.5) {};
\node[circle,fill=red,inner sep=1.5pt] at (4.5,5.5)  {};
\node[circle,fill=red,inner sep=1.5pt] at (7.5,5.5) {};
\node[circle,fill=red,inner sep=1.5pt] at (8.5,5.5)  {};
\node[circle,fill=red,inner sep=1.5pt] at (4.5,4.5)  {};
\node[circle,fill=red,inner sep=1.5pt] at (7.5,4.5)  {};
\node[circle,fill=red,inner sep=1.5pt] at (1.5,3.5) {};
\node[circle,fill=red,inner sep=1.5pt] at (4.5,3.5)  {};
\node[circle,fill=red,inner sep=1.5pt] at (7.5,3.5) {};
\node[circle,fill=red,inner sep=1.5pt] at (10.5,3.5){};
\node[circle,fill=red,inner sep=1.5pt] at (2.5,2.5)  {};
\node[circle,fill=red,inner sep=1.5pt] at (4.5,2.5)  {};
\node[circle,fill=red,inner sep=1.5pt] at (6.5,2.5)  {};
\node[circle,fill=red,inner sep=1.5pt] at (7.5,2.5)  {};
\node[circle,fill=red,inner sep=1.5pt] at (9.5,2.5)  {};
\node[circle,fill=red,inner sep=1.5pt] at (4.5,1.5)  {};
\node[circle,fill=red,inner sep=1.5pt] at (6.5,1.5)  {};
\node[circle,fill=red,inner sep=1.5pt] at (7.5,1.5)  {};
\fi
\node at (5.5,2.5) { $A$ };
\node at (2.5,8.5) { $B$ };
\node at (9.5,8.5) { $C$ };
\end{tikzpicture}
}

\medskip
\scalebox{\mboxscale}{
\begin{tikzpicture}[scale=\mmboxscale]
\draw[gray] (0,0) grid (12,9);
\filldraw[fill=blue!50,draw=black] (0,8) rectangle (1,9);
\filldraw[fill=blue!50,draw=black] (1,8) rectangle (2,9);
\filldraw[fill=blue!50,draw=black] (2,8) rectangle (3,9);
\filldraw[fill=blue!50,draw=black] (3,8) rectangle (4,9);
\filldraw[fill=blue!50,draw=black] (4,8) rectangle (5,9);
\filldraw[fill=blue!50,draw=black] (5,8) rectangle (6,9);
\filldraw[fill=blue!50,draw=black] (6,8) rectangle (7,9);
\filldraw[fill=blue!50,draw=black] (7,8) rectangle (8,9);
\filldraw[fill=blue!50,draw=black] (8,8) rectangle (9,9);
\filldraw[fill=blue!50,draw=black] (9,8) rectangle (10,9);
\filldraw[fill=blue!50,draw=black] (10,8) rectangle (11,9);
\filldraw[fill=blue!50,draw=black] (11,8) rectangle (12,9);
\filldraw[fill=blue!50,draw=black] (0,7) rectangle (1,8);
\filldraw[fill=blue!50,draw=black] (1,7) rectangle (2,8);
\filldraw[fill=blue!50,draw=black] (2,7) rectangle (3,8);
\filldraw[fill=blue!50,draw=black] (3,7) rectangle (4,8);
\filldraw[fill=blue!50,draw=black] (4,7) rectangle (5,8);
\filldraw[fill=blue!50,draw=black] (5,7) rectangle (6,8);
\filldraw[fill=blue!50,draw=black] (6,7) rectangle (7,8);
\filldraw[fill=blue!50,draw=black] (7,7) rectangle (8,8);
\filldraw[fill=blue!50,draw=black] (8,7) rectangle (9,8);
\filldraw[fill=blue!50,draw=black] (9,7) rectangle (10,8);
\filldraw[fill=blue!50,draw=black] (10,7) rectangle (11,8);
\filldraw[fill=blue!50,draw=black] (11,7) rectangle (12,8);
\filldraw[fill=blue!50,draw=black] (0,6) rectangle (1,7);
\filldraw[fill=blue!50,draw=black] (1,6) rectangle (2,7);
\filldraw[fill=blue!50,draw=black] (2,6) rectangle (3,7);
\filldraw[fill=blue!50,draw=black] (3,6) rectangle (4,7);
\filldraw[fill=blue!50,draw=black] (4,6) rectangle (5,7);
\filldraw[fill=blue!50,draw=black] (5,6) rectangle (6,7);
\filldraw[fill=blue!50,draw=black] (6,6) rectangle (7,7);
\filldraw[fill=blue!50,draw=black] (7,6) rectangle (8,7);
\filldraw[fill=blue!50,draw=black] (8,6) rectangle (9,7);
\filldraw[fill=blue!50,draw=black] (9,6) rectangle (10,7);
\filldraw[fill=blue!50,draw=black] (10,6) rectangle (11,7);
\filldraw[fill=blue!50,draw=black] (11,6) rectangle (12,7);
\filldraw[fill=blue!50,draw=black] (6,5) rectangle (7,6);
\filldraw[fill=blue!50,draw=black] (7,5) rectangle (8,6);
\filldraw[fill=blue!50,draw=black] (8,5) rectangle (9,6);
\filldraw[fill=blue!50,draw=black] (9,5) rectangle (10,6);
\filldraw[fill=blue!50,draw=black] (10,5) rectangle (11,6);
\filldraw[fill=blue!50,draw=black] (11,5) rectangle (12,6);
\filldraw[fill=blue!50,draw=black] (0,4) rectangle (1,5);
\filldraw[fill=blue!50,draw=black] (1,4) rectangle (2,5);
\filldraw[fill=blue!50,draw=black] (2,4) rectangle (3,5);
\filldraw[fill=blue!50,draw=black] (7,4) rectangle (8,5);
\filldraw[fill=blue!50,draw=black] (9,4) rectangle (10,5);
\filldraw[fill=blue!50,draw=black] (10,4) rectangle (11,5);
\filldraw[fill=blue!50,draw=black] (11,4) rectangle (12,5);
\filldraw[fill=blue!50,draw=black] (0,3) rectangle (1,4);
\filldraw[fill=blue!50,draw=black] (1,3) rectangle (2,4);
\filldraw[fill=blue!50,draw=black] (2,3) rectangle (3,4);
\filldraw[fill=blue!50,draw=black] (3,3) rectangle (4,4);
\filldraw[fill=blue!50,draw=black] (4,3) rectangle (5,4);
\filldraw[fill=blue!50,draw=black] (5,3) rectangle (6,4);
\filldraw[fill=blue!50,draw=black] (6,3) rectangle (7,4);
\filldraw[fill=blue!50,draw=black] (9,3) rectangle (10,4);
\filldraw[fill=blue!50,draw=black] (10,3) rectangle (11,4);
\filldraw[fill=blue!50,draw=black] (11,3) rectangle (12,4);
\filldraw[fill=blue!50,draw=black] (0,2) rectangle (1,3);
\filldraw[fill=blue!50,draw=black] (1,2) rectangle (2,3);
\filldraw[fill=blue!50,draw=black] (2,2) rectangle (3,3);
\filldraw[fill=blue!50,draw=black] (3,2) rectangle (4,3);
\filldraw[fill=blue!50,draw=black] (4,2) rectangle (5,3);
\filldraw[fill=blue!50,draw=black] (5,2) rectangle (6,3);
\filldraw[fill=blue!50,draw=black] (6,2) rectangle (7,3);
\filldraw[fill=blue!50,draw=black] (0,1) rectangle (1,2);
\filldraw[fill=blue!50,draw=black] (1,1) rectangle (2,2);
\filldraw[fill=blue!50,draw=black] (2,1) rectangle (3,2);
\filldraw[fill=blue!50,draw=black] (3,1) rectangle (4,2);
\filldraw[fill=blue!50,draw=black] (4,1) rectangle (5,2);
\filldraw[fill=blue!50,draw=black] (5,1) rectangle (6,2);
\filldraw[fill=blue!50,draw=black] (6,1) rectangle (7,2);
\filldraw[fill=blue!50,draw=black] (7,1) rectangle (8,2);
\filldraw[fill=blue!50,draw=black] (8,1) rectangle (9,2);
\filldraw[fill=blue!50,draw=black] (9,1) rectangle (10,2);
\filldraw[fill=blue!50,draw=black] (10,1) rectangle (11,2);
\filldraw[fill=blue!50,draw=black] (11,1) rectangle (12,2);
\filldraw[fill=blue!50,draw=black] (0,0) rectangle (1,1);
\filldraw[fill=blue!50,draw=black] (1,0) rectangle (2,1);
\filldraw[fill=blue!50,draw=black] (2,0) rectangle (3,1);
\filldraw[fill=blue!50,draw=black] (3,0) rectangle (4,1);
\filldraw[fill=blue!50,draw=black] (4,0) rectangle (5,1);
\filldraw[fill=blue!50,draw=black] (5,0) rectangle (6,1);
\filldraw[fill=blue!50,draw=black] (6,0) rectangle (7,1);
\filldraw[fill=blue!50,draw=black] (8,0) rectangle (9,1);
\filldraw[fill=blue!50,draw=black] (9,0) rectangle (10,1);
\filldraw[fill=blue!50,draw=black] (10,0) rectangle (11,1);
\filldraw[fill=blue!50,draw=black] (11,0) rectangle (12,1);
\iftrees
\draw[*-,red,very thick] (3.5,7.5) -- (3.5,6.5) -- (3.5,5.75);
\draw[*-,red,very thick] (5.75,7.75) -- (5.75,6.75) -- (5.75,6);
\draw[*-,red,very thick] (1.5,6.5) -- (1.5,5.75);
\draw[*-,red,very thick] (2.5,6.5) -- (2.5,5.5) -- (3.5,5.5) -- (4.25,5.5);
\draw[*-,red,very thick] (3.75,6.75) -- (3.75,5.75) -- (4.5,5.75);
\draw[*-,red,very thick] (4.25,6.25) -- (4.25,5.25) -- (3.25,5.25) -- (2.25,5.25) -- (1.5,5.25);
\draw[*-,red,very thick] (5.5,6.5) -- (5.5,5.5) -- (5.5,4.75);
\draw[*-,red,very thick] (6.5,6.5) -- (6.5,5.75);
\draw[*-,red,very thick] (7.5,6.5) -- (7.5,5.75);
\draw[*-,red,very thick] (6.5,5.5) -- (5.5,5.5) -- (4.5,5.5) -- (3.5,5.5) -- (2.5,5.5) -- (1.75,5.5);
\draw[*-,red,very thick] (7.5,5.5) -- (6.75,5.5);
\draw[*-,red,very thick] (8.75,5.75) -- (8.75,4.75) -- (8.75,4);
\draw[*-,red,very thick] (1.5,4.5) -- (2.25,4.5);
\draw[*-,red,very thick] (2.5,4.5) -- (3.5,4.5) -- (4.5,4.5) -- (5.5,4.5) -- (6.5,4.5) -- (7.25,4.5);
\draw[*-,red,very thick] (7.75,4.75) -- (7.75,3.75) -- (7.75,4.5);
\draw[*-,red,very thick] (9.5,4.5) -- (9.5,3.75);
\draw[-,red,very thick] (9.5,4.5) -- (8.5,4.5) -- (7.5,4.5) -- (6.5,4.5) -- (5.5,4.5) -- (4.75,4.5);
\draw[*-,red,very thick] (10.25,4.25) -- (9.25,4.25) -- (8.25,4.25) -- (7.25,4.25) -- (6.5,4.25);
\draw[*-,red,very thick] (1.5,3.5) -- (1.5,4.25);
\draw[*-,red,very thick] (2.5,3.5) -- (2.5,4.25);
\draw[*-,red,very thick] (3.5,3.5) -- (3.5,4.5) -- (4.5,4.5) -- (5.5,4.5) -- (6.5,4.5) -- (7.25,4.5);
\draw[*-,red,very thick] (4.75,3.75) -- (4.75,4.75) -- (5.75,4.75) -- (6.75,4.75) -- (7.5,4.75);
\draw[*-,red,very thick] (5.25,3.25) -- (5.25,4.25) -- (6.25,4.25) -- (7,4.25);
\draw[*-,red,very thick] (6.5,3.5) -- (7.25,3.5);
\draw[-,red,very thick] (6.5,3.5) -- (6.5,4.5) -- (5.5,4.5) -- (4.75,4.5);
\draw[*-,red,very thick] (9.5,3.5) -- (8.75,3.5);
\draw[*-,red,very thick] (3.5,2.5) -- (3.5,3.5) -- (3.5,4.25);
\draw[*-,red,very thick] (5.75,2.75) -- (5.75,3.75) -- (5.75,4.5);
\draw[*-,red,very thick] (6.25,2.25) -- (7.25,2.25) -- (8.25,2.25) -- (9.25,2.25) -- (10,2.25);
\draw[*-,red,very thick] (7.5,1.5) -- (7.5,2.5) -- (8.5,2.5) -- (9.5,2.5) -- (10.25,2.5);
\draw[*-,red,very thick] (8.75,1.75) -- (8.75,2.75) -- (8.75,3.5);
\draw[*-,red,very thick] (9.5,1.5) -- (9.5,2.25);
\else
\node[circle,fill=red,inner sep=1.5pt] at (3.5,7.5) {} ;
\node[circle,fill=red,inner sep=1.5pt] at (5.5,7.5) {} ;
\node[circle,fill=red,inner sep=1.5pt] at (1.5,6.5) {} ;
\node[circle,fill=red,inner sep=1.5pt] at (2.5,6.5) {} ;
\node[circle,fill=red,inner sep=1.5pt] at (3.5,6.5) {} ;
\node[circle,fill=red,inner sep=1.5pt] at (4.5,6.5) {} ;
\node[circle,fill=red,inner sep=1.5pt] at (5.5,6.5) {} ;
\node[circle,fill=red,inner sep=1.5pt] at (6.5,6.5) {} ;
\node[circle,fill=red,inner sep=1.5pt] at (7.5,6.5) {} ;
\node[circle,fill=red,inner sep=1.5pt] at (6.5,5.5) {} ;
\node[circle,fill=red,inner sep=1.5pt] at (7.5,5.5) {} ;
\node[circle,fill=red,inner sep=1.5pt] at (8.5,5.5) {} ;
\node[circle,fill=red,inner sep=1.5pt] at (1.5,4.5) {} ;
\node[circle,fill=red,inner sep=1.5pt] at (2.5,4.5) {} ;
\node[circle,fill=red,inner sep=1.5pt] at (7.5,4.5) {} ;
\node[circle,fill=red,inner sep=1.5pt] at (9.5,4.5) {} ;
\node[circle,fill=red,inner sep=1.5pt] at (9.5,4.5) {} ;
\node[circle,fill=red,inner sep=1.5pt] at (10.5,4.5) {} ;
\node[circle,fill=red,inner sep=1.5pt] at (1.5,3.5) {} ;
\node[circle,fill=red,inner sep=1.5pt] at (2.5,3.5) {} ;
\node[circle,fill=red,inner sep=1.5pt] at (3.5,3.5) {} ;
\node[circle,fill=red,inner sep=1.5pt] at (4.5,3.5) {} ;
\node[circle,fill=red,inner sep=1.5pt] at (5.5,3.5) {} ;
\node[circle,fill=red,inner sep=1.5pt] at (6.5,3.5) {} ;
\node[circle,fill=red,inner sep=1.5pt] at (6.5,3.5) {} ;
\node[circle,fill=red,inner sep=1.5pt] at (9.5,3.5) {} ;
\node[circle,fill=red,inner sep=1.5pt] at (3.5,2.5) {} ;
\node[circle,fill=red,inner sep=1.5pt] at (5.5,2.5) {} ;
\node[circle,fill=red,inner sep=1.5pt] at (6.5,2.5) {} ;
\node[circle,fill=red,inner sep=1.5pt] at (7.5,1.5) {} ;
\node[circle,fill=red,inner sep=1.5pt] at (8.5,1.5) {} ;
\node[circle,fill=red,inner sep=1.5pt] at (9.5,1.5) {} ;
\fi
\node at (11.5,2.5) { $A$ };
\node at (0.5,5.5) { $B$ };
\node at (7.5,0.5) { $C$ };
\end{tikzpicture}
}

\caption{\small Movable blocks (marked by a red dot) in the three elementary gadgets.}
  \label{fig:ann1}
\end{figure}
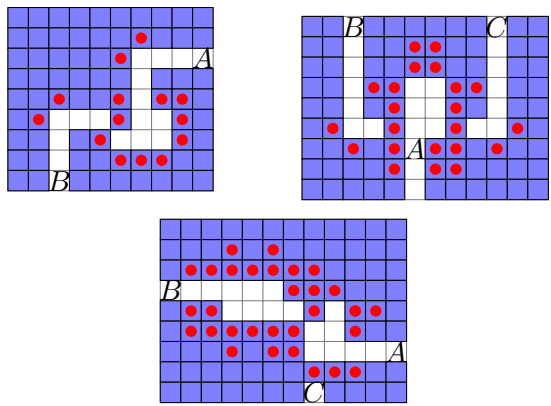

\section{Open questions}

We have investigated some variants of pull-like block moving games,
but left open the complexity of Pull-$k$ and PullPull-$k$ variants. A
problem in establishing the hardness of the corresponding
Push-variants has been to prevent neighboring gadgets from
interfering, in particular for $k=*$~\cite{Demaine.etal/2003}. This is
not the case here, since gadgets can be easily isolated by inserting
holes between them. The more difficult part is to give a succinct
argument that pulling $k$ blocks at once never opens new paths, but we
conjecture that this is possible. Another open question is the
$\PSPACE$-hardness of these problems.

A interesting restriction could be to study \emph{handles}, where a
box can be pulled only in a direction, if the corresponding side has a
handle. Hardness of the problem where the handles can be specified as
part of the input follows from the theorems above. Are restricted
variants easier, for example the variant where a box has either only
up-down handles or left-right handles? Another interesting variant is
to combine push and pull moves when all blocks are movable. This may
allow us to answer positively the open question from
\cite{Demaine.etal/2003}, if there exists an ``interesting'' but
tractable block-moving puzzle.

\bibliographystyle{plainnat}
\bibliography{paper,complexity,comments}

\end{document}